\documentclass[a4paper, 10pt, notitlepage]{article}

\usepackage{amsmath, amsthm, amsfonts, tikz, hyperref, graphics, authblk, mathdots}
\usepackage[font=small, margin=10mm, labelfont=bf]{caption}
\DeclareMathOperator{\Tr}{Tr}
\DeclareMathOperator{\e}{e}
\DeclareMathOperator{\im}{i}
\DeclareMathOperator{\di}{d}

\newtheorem{theorem}{Theorem}
\newtheorem{corollary}{Corollary}
\newtheorem{lemma}{Lemma}

\numberwithin{equation}{section}

\title{Random matrices and quantum spin chains}

\author[1]{J.P. Keating\thanks{j.p.keating@bristol.ac.uk}}
\author[1]{N. Linden\thanks{n.linden@bristol.ac.uk}}
\author[1]{H.J. Wells\thanks{huw.wells@bristol.ac.uk}}
\affil[1]{School of Mathematics, University of Bristol, Bristol, BS8 1TW, UK}

\begin{document}
\date{March 2014}
\maketitle

\begin{abstract}
Random matrix ensembles are introduced that respect the local tensor structure of Hamiltonians describing a chain of $n$ distinguishable spin-half particles with nearest-neighbour interactions.  We prove a central limit theorem for the density of states when $n \rightarrow\infty$, giving explicit bounds on the rate of approach to the limit.  Universality within a class of probability measures and the extension to more general interaction geometries are established.  The level spacing distributions of the Gaussian Orthogonal, Unitary and Symplectic Ensembles are observed numerically for the energy levels in these ensembles.
\footnotesize{\newline\newline AMS subject classification numbers: 60B20, 82B10, 81P45, 81Q50}
\footnotesize{\newline Keywords: quantum spin chains, density of states, spectral statistics, random matrix theory}
\end{abstract}

\begin {center} 
{\it Dedicated to Leonid Pastur, to mark his $75^{th}$ birthday}
\end{center}

\section{Introduction}
Random matrices were introduced into Physics by Wigner to model the statistical properties of many-body quantum systems, specifically heavy nuclei.  To simplify the model, Wigner assumed that each body (i.e.~each nucleon, in the example of nuclei) interacts equally strongly with all of the others.  The Hamiltonian matrices therefore have no structure other than that dictated by global symmetries of the Hilbert space, such a time reversal, when these are present.  This philosophy underpins Dyson's threefold way and its later extension by Altland and Zirnbauer.  In this respect, the usual random matrix ensembles do not reflect the many-body nature of the problems being modelled; indeed they apply to quantum chaotic systems with only a few degrees of freedom.  For an overview of this background, see, for example, related articles in  \cite{Handbook}. 

French and Wong \cite{French1970,French1971} and Bohigas and Flores \cite{Bohigas1971,Bohigas1971a} introduced random matrix models for systems of $n$ indistinguishable particles in which only  $k=2,3,\dots,n$ particle interactions are allowed, in the form of the Embedded Gaussian Orthogonal Ensembles (EGOE($k$)).  These ensembles, in particular their dependence on $k$, have for the most part been investigated numerically.  When $k=n$ they coincide with the standard Gaussian Orthogonal Ensemble (GOE), but when $k\ll n$ they exhibit interesting and distinctive features.

Our aim here is to continue this line of investigation by developing ensembles of random matrices that model quantum spin chains with nearest neighbour interactions.  These ensembles are similar to those previously considered by Pi\v{z}orn, Prosen, Mossmann and Seligman \cite{Prosen2007}, who studied the spectral gap between the ground and first excited states, observing a transition from the nearest neighbour statistics of the Gaussian Unitary Ensemble to that of a Poisson point process.  We will focus primarily on calculating the ensemble-averaged density of states rigorously.  In particular, we show that when $n \rightarrow \infty$ the density obeys a central limit theorem, and we obtain explicit bounds on the rate of approach to this limit.  These ensembles therefore differ significantly from the usual Wigner ensembles and invariant ensembles of random matrix theory, for which the semicircle law and its generalizations hold.  We show that this result is universal for a class of probability measures defining the ensembles and for a class of geometries determining the interactions in the spin system.  We also consider anti-unitary symmetries of the ensembles and show how these influence the spectral statistics by computing nearest-neighbour spacing distributions numerically. 

The central limit theorem we prove for the ensemble-averaged density of states is in agreement with recent numerical computations of Atas and Bogomolny \cite{Atas-Bogomolny2013} and with our own reported here.  In particular the rate of approach to the limit that we are able to establish is consistent with these numerical results.  The proof of the central limit theorem relies on splitting the spin-chain into sub-chains within which components of the Hamiltonians that constitute the matrix ensembles commute with each other.  It is thus reminiscent of an interesting calculation of Hartmann, Mahler and Hess \cite{Mahler,Mahler2005}, who considered an arbitrary fixed Hamiltonian of a qubit chain and proved, under general conditions, that the energy distribution of almost any pure product state, over the energy eigenbasis, weakly approaches a Gaussian.

The range of spin chain models that are analytically diagonalisable is limited. The Jordan-Wigner transformation \cite{Mattis1961} and Bethe ansatz \cite{Muller2004} allow the exact eigenstates and eigenenergies of spin chain models to be determined in special cases, but for the majority of generic quantum spin chain models such techniques do not exist.  Specific examples of generic chains have been realised by many authors \cite{Santos2011,Millis2004,Deguchi2004,Angles1992}.  Many of these models share the property that in the large-chain limit their eigenvalues, within symmetry subspaces, exhibit level repulsion similar to that of the canonical random matrix ensembles.  Our findings concerning the spectral statistics are consistent with this and emphasize the important role played by anti-unitary symmetries.  

One of our main motivations in this work is to investigate the statistical properties of many-body systems from the perspective of quantum information theory.  Quantum spin chains are canonical models in this context; for example, they have been used to study various measures of entanglement and how these are influenced by phase transitions \cite{Osborne2002,Osterloh2002,Vidal2003}, the efficiency of quantum state transfer \cite{Linden2004,Bose2003,Bose2005}, etc.  Typically, but not exclusively (see, for example \cite{Keat-Mezz2004,Keat-Mezz2005,Keating-Linden2007,Calabrese2004}) calculations have been restricted to the ground states of exactly solvable models.  We hope that introducing random matrix models will provide a new approach to these questions that will enable the excited states of non-integrable systems to be studied.  The extent to which ensemble averages describe the features of individual systems in the limit when $n \rightarrow \infty$ -- that is, the extent to which the ensembles exhibit ergodicity in this limit -- is then a key issue.  For the density of states we shall explore this in a related paper \cite{KLW2014}, where we prove that the central limit theorem established here for ensembles holds for a general class of given Hamiltonians in line with previous conjectures, and where we examine the connections to quantum information theory more generally.

This paper is arranged as follows.  In Section \ref{model}, the basic Gaussian random matrix model for quantum spin chains is introduced.  In Section \ref{results}, the central limit theorem for the density of states is proved and an explicit bound on the characteristic function of the associated probability measure is established.  In Sections \ref{uni}, \ref{geo} and $\ref{local}$ we consider the universality of this result.  Level spacing statistics and the role played by anti-unitary symmetries are investigated in Section \ref{level}.   Finally, in Section \ref{dis}, we conclude with a brief discussion of some open problems.

\section{Basic random matrix ensemble}\label{model}
Let the $2^n\times2^n$ Hermitian random matrix $H_n$ be defined as
\begin{equation}\label{ensemble}
  H_n=\sum_{j=1}^n\sum_{a,b=1}^3\alpha_{a,b,j}\sigma_{  j  }^{(a)}\sigma_{  j+1  }^{(b)}
\end{equation}
for integers $n\geq2$, where the coefficients $\alpha_{a,b,j}$ are $9n$ independent normally distributed random variables with zero mean and variance $\frac{1}{9n}$, and
\begin{equation}
  \sigma_{  j  }^{(a)}=I_2^{\otimes(j-1)}\otimes\sigma^{(a)}\otimes I_2^{\otimes(n-j)}
\end{equation}
where $\sigma^{(1)}$, $\sigma^{(2)}$ and $\sigma^{(3)}$ are the $2\times 2$ Pauli matrices and $I_2\equiv\sigma^{(0)}$ is the $2\times2$ identity matrix:
\begin{equation}\label{Pauli}
	\sigma^{(0)}=\begin{pmatrix}1&0\\0&1\end{pmatrix}\qquad
	\sigma^{(1)}=\begin{pmatrix}0&1\\1&0\end{pmatrix}\qquad
	\sigma^{(2)}=\begin{pmatrix}0&-\im\\\im&0\end{pmatrix}\qquad
	\sigma^{(3)}=\begin{pmatrix}1&0\\0&-1\end{pmatrix}
\end{equation}
The labelling is cyclic so that $\sigma_{n+1}^{(a)}\equiv\sigma_{1}^{(a)}$.  The random matrix $H_n$ acts on the Hilbert space of $n$ distinguishable qubits, $\left(\mathbb{C}^2\right)^{\otimes n}$, which is the $n$-fold tensor product of the individual qubit Hilbert spaces $\mathbb{C}^2$.

The ensemble of Hamiltonians $H_n$  describes a ring of qubits that interact with their nearest neighbours.  It will be seen later that the results presented here also apply to other related ensembles;  for example ensembles with different probability measures, that may include `local' terms proportional to $\sigma_j^{(a)}$, and that may be based on more elaborate interaction geometries. We note the similarity to the Hamiltonians studied numerically by Pizorn, Prosen, Mossmann and Seligman \cite{Prosen2007}.

The density of states probability measure $\di\mu_n^{(DOS)}$ for this ensemble is a probability measure on the real line, induced from the matrix $H_n$, so that its integral over any interval is the expected proportion of the eigenvalues of $H_n$ lying within that interval.  Formally this may be expressed as
\begin{equation}
  \di\mu_n^{(DOS)}(\lambda)=\left\langle\frac{1}{2^n}\sum_{k=1}^{2^n}\delta(\lambda-\lambda_k)\right\rangle\di\lambda
\end{equation}
where $\lambda_k$ are the eigenvalues of $H_n$ and the average (denoted by the angular brackets) is over the ensemble.  More precisely, the probability measure $\di\mu_n(H_n)$ can be uniformly re-parametrised in terms of $2^n$ unordered real parameters (eigenvalues of $H_n$) and $4^n-2^n$ additional real parameters (eigen-directions of $H_n$ ).  Integrating out all but one of the variables associated to the eigenvalues produces a measure equivalent to $\di\mu_n^{(DOS)}$.

The characteristic function $\psi_n(t)$ associated with $\di\mu_n^{(DOS)}$ is defined to be
\begin{equation}\label{char}
  \psi_n(t)=\mathbb{E}_{\mu_n^{(DOS)}}\left(\e^{\im t\lambda}\right)=\int\e^{\im t\lambda}\di\mu_n^{(DOS)}(\lambda)=\left\langle\frac{1}{2^n}\Tr\e^{\im tH_n}\right\rangle
\end{equation}

\section{Central limit theorem}\label{results}
Our first result concerns the convergence of the characteristic function $\psi_n(t)$ associated with the density of states probability measure $\di\mu_n^{(DOS)}$ for the ensemble of matrices $H_n$:
\begin{theorem}[Convergence of the characteristic function $\psi_n(t)$]\label{DOS}
  For $n\in2\mathbb{N}_+$, the characteristic function $\psi_n(t)$ converges pointwise to the characteristic function of a standard normal random variable as $n\to\infty$, specifically,
  \begin{equation}
    \left|\psi_n(t)-\e^{-\frac{t^2}{2}}\right|\leq t^2s_n^2\sqrt{n}\left(36\sqrt{2}+81\right)=\frac{t^2\left(4\sqrt{2}+9\right)}{\sqrt{n}}
  \end{equation}
  where $s_n^2=\frac{1}{9n}$ is the variance of the random variables $\alpha_{a,b,j}$.
\end{theorem}

L\'{e}vy's continuity theorem \cite{Bill} states that a sequence of probability measures weakly converge to a limiting probability measure if, and only if, the associated sequence of characteristic functions converge pointwise to the characteristic function associated with that limiting probability measure.  This leads immediately to the following corollary:

\begin{corollary}[Weak convergence of the density of states probability measure]
  The density of states probability measure $\di\mu_{2n}^{(DOS)}$ weakly converges to that of a standard normal distribution.  That is, for any $x\in\mathbb{R}$,
  \begin{equation}
    \lim_{n\to\infty}\int_{-\infty}^x\di\mu_{2n}^{(DOS)}=\frac{1}{\sqrt{2\pi}}\int_{-\infty}^x\e^{-\frac{\lambda^2}{2}}\di\lambda
  \end{equation}
\end{corollary}

Appendix \ref{numDOS} contains the results of a numerical simulation that illustrates the nature of this convergence to the limit.  Note that the bound we are able to determine is much weaker than one has for the standard Wigner ensembles, because in our case the mean level separation is exponentially small in $n$.  The theorem here is presented for even values of $n$ for simplicity of the proof.  The extension to odd values of $n$ and more general interaction geometries is made in Section \ref{geo}.

The proof of Theorem \ref{DOS} will now be given:  The key idea is to treat the terms in $H_n$ as independent commuting random variables, so that the characteristic function of the density of states for $H_n$ can be approximated by a product of characteristic functions of the corresponding densities for each of the terms in $H_n$.  The extent to which this fails to be true exactly leads to an error which can be bounded.
\begin{proof}
First, let $n$ be even and
\begin{equation}\label{A}
  A=\sum_{\genfrac{}{}{0pt}{}{j=1}{j\text{ even}}}^n\sum_{a,b=1}^3\alpha_{a,b,j}\sigma_{  j  }^{(a)}\sigma_{  j+1  }^{(b)}\qquad\qquad
  B=\sum_{\genfrac{}{}{0pt}{}{j=1}{j\text{ odd}}}^n\sum_{a,b=1}^3\alpha_{a,b,j}\sigma_{  j  }^{(a)}\sigma_{  j+1  }^{(b)}
\end{equation}
and
\begin{equation}\label{A_i}
  A_{3(b-1)+a}=\sum_{\genfrac{}{}{0pt}{}{j=1}{j\text{ even}}}^n\alpha_{a,b,j}\sigma_{  j  }^{(a)}\sigma_{  j+1  }^{(b)}\qquad\qquad
  B_{3(b-1)+a}=\sum_{\genfrac{}{}{0pt}{}{j=1}{j\text{ odd}}}^n\alpha_{a,b,j}\sigma_{  j  }^{(a)}\sigma_{  j+1  }^{(b)}
\end{equation}
so that all the terms within each sum for each operator $A_k$ and $B_k$ commute and
\begin{equation}
  H_n=A+B=\sum_{k=1}^9\big(A_k+B_k\big)
\end{equation}

Now, let $\phi_n(t)$ be the characteristic function formed from the product of the characteristic functions associated to the density of states probability measure for each random matrix $\alpha_{a,b,j}\sigma_{  j  }^{(a)}\sigma_{  j+1  }^{(b)}$.  That is, by a calculation analogous to that in equation (\ref{char}) for each factor,
\begin{equation}\label{characteristic}
  \phi_n(t)=\prod_{j=1}^n\prod_{a=1}^3\prod_{b=1}^3\left\langle\frac{1}{2^n}\Tr\e^{\im t\alpha_{a,b,j}\sigma_{  j  }^{(a)}\sigma_{  j+1  }^{(b)}}\right\rangle
\end{equation}
The trace in this and subsequent expressions is over all $n$ sites.  It will now be shown that this is exactly equal to
\begin{equation}\label{characteristic2}
  \left\langle\frac{1}{2^n}\Tr\left(\prod_{k=1}^9\e^{\im tA_k}\e^{\im tB_k}\right)\right\rangle
\end{equation}
First, the trace and the average will be interchanged in (\ref{characteristic}).  Given any operator $M$ with elements $M_{jk}$ in some arbitrary basis, it then follows from the linearity of the trace that $\langle\Tr M\rangle=\sum\langle M_{jj}\rangle=\Tr\langle M\rangle$ so that
\begin{equation}\label{characteristic3}
  \phi_n(t)=\prod_{j=1}^n\prod_{a=1}^3\prod_{b=1}^3\frac{1}{2^n}\Tr\left\langle\e^{\im t\alpha_{a,b,j}\sigma_{  j  }^{(a)}\sigma_{  j+1  }^{(b)}}\right\rangle
\end{equation}
As the square of any Pauli matrix is the identity, by Taylor expanding the exponential 
\begin{equation}\label{factorAvg}
  \left\langle\e^{\im t\alpha_{a,b,j}\sigma_{  j  }^{(a)}\sigma_{  j+1  }^{(b)}}\right\rangle
  =\left\langle\cos\left(t\alpha_{a,b,j}\right)\right\rangle I_{2^n}+\im\left\langle\sin\left(t\alpha_{a,b,j}\right)\right\rangle\sigma_{  j  }^{(a)}\sigma_{  j+1  }^{(b)}
\end{equation}
where the second term is zero by symmetry of the measure.  The right hand side of (\ref{factorAvg}) is therefore proportional to the identity, allowing the product to be taken inside the trace, in the last expression for $\phi_n(t)$, to give
\begin{equation}
  \phi_n(t)=\frac{1}{2^n}\Tr\left(\prod_{j=1}^n\prod_{a=1}^3\prod_{b=1}^3\left\langle\e^{\im t\alpha_{a,b,j}\sigma_{  j  }^{(a)}\sigma_{  j+1  }^{(b)}}\right\rangle\right)\end{equation}
Again since the factors in the product are proportional to the identity, their order is irrelevant.  The factors are also all statistically independent which allows the product of the average of each factor to be written as the average of the product of the factors.  Therefore
\begin{equation}
  \phi_n(t)=\frac{1}{2^n}\Tr\left\langle\prod_{k=1}^9\e^{\im tA_k}\e^{\im tB_k}\right\rangle
\end{equation}
as all the terms within each sum for each operator $A_k$ and $B_k$ commute.  Finally, swapping the trace and the average results in the expression for $\phi_n(t)$ claimed in (\ref{characteristic2}).

The average in (\ref{factorAvg}) is exactly computable in the case of a Gaussian measure and is equal to $\e^{-\frac{s_n^2t^2}{2}}I_{2^n}$.  From the calculations above it then follows that $\phi_n(t)=\e^{-\frac{t^2}{2}}$.

It now only remains to show that $\psi_n(t)$ and $\phi_n(t)$ satisfy the bound claimed.  To this end, let the $18$ matrices $\{A_k,B_k\}_{k=1}^9$ be relabelled by $\{Q_k\}_{k=1}^{18}$ such that $Q_k=A_k$ if $1\leq k\leq9$ and $Q_k=B_{k-9}$ if $10\leq k\leq18$ and let $Q_{19}=0$.  The telescoping sum
\begin{equation}
  \frac{1}{2^n}\Tr\left(\e^{\im tH_n}-\prod_{k=1}^{18}\e^{\im tQ_k}\right)
  =\sum_{s=2}^{18}\frac{1}{2^n}\Tr\left(\left(\e^{\im t\sum_{k=1}^sQ_k}-\e^{\im t\sum_{j=1}^{s-1}Q_j}\e^{\im tQ_s}\right)\prod_{l=s+1}^{19}\e^{\im tQ_l}\right)
\end{equation}
and the triangle inequality then give the bound
\begin{align}\label{CS}
  |\psi_n(t)-\phi_n(t)|&=\left|\left\langle\frac{1}{2^n}\Tr\e^{\im tH_n}\right\rangle-\left\langle\frac{1}{2^n}\Tr\left(\prod_{k=1}^{19}\e^{\im tQ_k}\right)\right\rangle\right|\nonumber\\
  &\leq\sum_{s=2}^{18}\left\langle\left|\frac{1}{2^n}\Tr\left(\left(\e^{\im t\sum_{k=1}^sQ_k}-\e^{\im t\sum_{j=1}^{s-1}Q_j}\e^{\im tQ_s}\right)\prod_{l=s+1}^{19}\e^{\im tQ_l}\right)\right|\right\rangle
\end{align}
as the order of the exponential factors in (\ref{characteristic2}) has been seen to be arbitrary.  The integral identity for any $2^n\times2^n$ Hermitian matrices $X$ and $Y$,
\begin{align}
  \e^{\im t(X+Y)}-\e^{\im tX}\e^{\im tY}=\int_0^1\int_0^1t^2s\e^{\im (1-s)t(X+Y)}\e^{\im (1-r)stX}[X,Y]\e^{\im rstX}\e^{\im stY}\di r\di s
\end{align}
(see Appendix \ref{Int}) and the Cauchy Schwartz inequality for any $2^n\times2^n$ matrix $M$,
\begin{equation}
  \left|\Tr M\right|^2=\left|\sum_{j=1}^{2^n}M_{jj}\right|^2\leq\sum_{j=1}^{2^n}|M_{jj}|^2\sum_{k=1}^{2^n}|1|^2\leq2^n\Tr\left(MM^\dagger\right)
\end{equation}
can be used to bound each of the $17$ terms in (\ref{CS}).  For any $2^n\times2^n$ unitary matrix $U$, this pair of inequalities along with the triangle inequality yield
\begin{align}
  &\left|\frac{1}{2^n}\Tr\left(\left(\e^{\im t(X+Y)}-\e^{\im tX}\e^{\im tY}\right)U\right)\right|\nonumber\\
  &\qquad\qquad\leq\int_0^1\int_0^1t^2s\left|\frac{1}{2^n}\Tr\left(\e^{\im (1-s)t(X+Y)}\e^{\im (1-r)stX}[X,Y]\e^{\im rstX}\e^{\im stY}U\right)\right|\di r\di s\nonumber\\
  &\qquad\qquad\qquad\qquad\leq\frac{t^2}{2}\sqrt{\frac{1}{2^n}\Tr\left([X,Y][X,Y]^\dagger\right)}
\end{align}
so that, for the subadditive matrix norm $\|M\|=\sqrt{\frac{1}{2^n}\Tr\left(MM^\dagger\right)}$,
\begin{equation}\label{bound}
 |\psi_n(t)-\phi_n(t)|\leq\frac{t^2}{2}\sum_{s=2}^{18}\left\langle\left\|\left[\sum_{j=1}^{s-1}Q_j,Q_s\right]\right\|\right\rangle\leq\frac{t^2}{2}\sum_{k<{k^\prime}}\left\langle\|[Q_k,Q_{k^\prime}]\|\right\rangle
\end{equation}

The averages of the norms $\|[A_k,A_{k^\prime}]\|$, $\|[B_k,B_{k^\prime}]\|$ and $\|[A_k,B_{k^\prime}]\|$ must now be calculated.  In the case of $\|[A_k,A_{k^\prime}]\|$ with the index $k=3(b-1)+a$ as in (\ref{A_i}), by definition,
\begin{equation}
  [A_k,A_{k^\prime}]=\sum_{\genfrac{}{}{0pt}{}{j=1}{j\text{ even}}}^n\sum_{\genfrac{}{}{0pt}{}{j^\prime=1}{j^\prime\text{ even}}}^n\alpha_{a,b,j}\alpha_{a^\prime,b^\prime,j^\prime}\left[\sigma_{  j  }^{(a)}\sigma_{  j+1  }^{(b)},\sigma_{  j^\prime  }^{( a^\prime )}\sigma_{  j^\prime+1  }^{( b^\prime )}\right]
\end{equation}
Terms for which $j\neq j^\prime$ here are zero, because in this case $\sigma_{  j  }^{(a)}\sigma_{  j+1  }^{(b)}$ and $\sigma_{  j^\prime  }^{( a^\prime )}\sigma_{  j^\prime+1  }^{( b^\prime )}$ commute.  The norm of $[A_k,A_{k^\prime}]$ is then, by definition,
\begin{equation}
  \left(\frac{1}{2^n}\Tr\left(\sum_{\genfrac{}{}{0pt}{}{j=1}{j\text{ even}}}^n\sum_{\genfrac{}{}{0pt}{}{l=1}{l\text{ even}}}^n\alpha_{a,b,j}\alpha_{a^\prime,b^\prime,j}\alpha_{a,b,l}\alpha_{a^\prime,b^\prime,l}\left[\sigma_{  j  }^{(a)}\sigma_{  j+1  }^{(b)},\sigma_{  j  }^{( a^\prime )}\sigma_{  j+1  }^{( b^\prime )}\right]\left[\sigma_{  l  }^{(a)}\sigma_{  l+1  }^{(b)},\sigma_{  l  }^{( a^\prime )}\sigma_{  l+1  }^{( b^\prime )}\right]^\dagger\right)\right)^{\frac{1}{2}}
\end{equation}
Jensen's Inequality allows the average of this quantity to be bounded by taking the average of all the terms inside the square root individually.  For $A_k\neq A_{k^\prime}$, $(a,b)\neq(a^\prime,b^\prime)$, so $\left\langle\alpha_{a,b,j}\alpha_{a^\prime,b^\prime,j}\alpha_{a,b,l}\alpha_{a^\prime,b^\prime,l}\right\rangle$ is non-zero only when $j=l$, by the symmetry of the average (hence $\langle\alpha_{a,b,j}\rangle=0$).  This non-zero value is precisely
\begin{equation}\label{variance}
  \left\langle\alpha_{a,b,j}\alpha_{a^\prime,b^\prime,j}\alpha_{a,b,j}\alpha_{a^\prime,b^\prime,j}\right\rangle
  =\left\langle\alpha_{a,b,j}^2\right\rangle\left\langle\alpha_{a^\prime,b^\prime,j}^2\right\rangle=s_n^4
\end{equation}
Furthermore, by directly applying the definition of the Pauli matrices and commutator,
\begin{align}	
  0\leq\frac{1}{2^n}\Tr\left(\left[\sigma_{  j  }^{(a)}\sigma_{  j+1  }^{(b)},\sigma_{  j  }^{( a^\prime )}\sigma_{  j+1  }^{( b^\prime )}\right]\left[\sigma_{  j  }^{(a)}\sigma_{  j+1  }^{(b)},\sigma_{  j  }^{( a^\prime )}\sigma_{  j+1  }^{( b^\prime )}\right]^\dagger\right)\leq4
\end{align}
so that
\begin{equation}
  \left\langle\|[A_k,A_{k^\prime}]\|\right\rangle\leq\sqrt{\sum_{\genfrac{}{}{0pt}{}{j=1}{j\text{ even}}}^n4s_n^4}=2s_n^2\sqrt{\frac{n}{2}}
\end{equation}

Similarly $\left\langle\|[B_k,B_{k^\prime}]\|\right\rangle\leq2s_n^2\sqrt{\frac{n}{2}}$ and $\left\langle\|[A_k,B_{k^\prime}]\|\right\rangle\leq2s_n^2\sqrt{n}$, so that equation (\ref{bound}) now gives that
\begin{align}
  |\psi_n(t)-\phi_n(t)|
  &\leq\frac{t^2}{2}\left(\sum_{k<{k^\prime}}\left\langle\|A_k,A_{k^\prime}\|\right\rangle+\sum_{k<{k^\prime}}\left\langle\|B_k,B_{k^\prime}\|\right\rangle+\sum_{k,{k^\prime}}\left\langle\|A_k,B_{k^\prime}\|\right\rangle\right)\nonumber\\
  &=t^2s_n^2\sqrt{n}\left(36\sqrt{2}+81\right)
\end{align}
completing the proof for even values of $n$ (see Section \ref{geo} for the extension to odd values).
\end{proof}

\subsection{Universality}\label{uni}
The proof of Theorem \ref{DOS} can be extended to hold for a range of distributions.

Lyapunov's central limit theorem \cite{Bill} states that if $x_1,x_2,\dots,x_{r(n)}$ are a collection of independent random variables (not necessarily identically distributed) for each value of $n\in\mathbb{N}$ separately, each with finite mean and variance and so that
\begin{equation}
  \sum_{j=1}^{r(n)}\left\langle x_j^2\right\rangle
  =1\qquad\qquad\lim_{n\to\infty}\sum_{j=1}^{r(n)}\left\langle\left|x_j-\left\langle x_j\right\rangle\right|^{2+\delta}\right\rangle=0\qquad\qquad\lim_{n\to\infty}r(n)=\infty
\end{equation}
for some $\delta\in\mathbb{N}$, then the sum of $x_j-\langle x_j\rangle$ converges weakly to a standard normal random variable.

The following theorem deals with a corresponding generalisation of the ensemble considered in the previous section:

\begin{theorem}[Universal convergence of the characteristic function $\psi_n(t)$]\label{universal}
  Let the $r(n)=9n$ real random variables $\alpha_{a,b,j}$ satisfy the conditions of Lyapunov's central limit theorem, be symmetric about zero and have a maximum variance of $s_n^2=o\left(\frac{1}{\sqrt{n}}\right)$.  Then for $n\in2\mathbb{N}$, the characteristic function $\psi_n(t)$ converges pointwise to the characteristic function of a standard normal random variable as $n\to\infty$.
\end{theorem}

It then follows immediately from the continuity theorem that $\di\mu_{2n}^{(DOS)}$ tends weakly to a standard normal distribution for this class of ensembles.

\begin{proof}
The proof follows that of Theorem \ref{DOS}, with the operators $A$, $B$, $A_k$ and $B_k$ defined identically.  The characteristic function $\phi_n(t)$, as defined in equation (\ref{characteristic}), can be rewritten by evaluating equations (\ref{characteristic3}) and (\ref{factorAvg}), as
\begin{equation}
  \phi_n(t)=\prod_{j=1}^n\prod_{a=1}^3\prod_{b=1}^3\left\langle\e^{\im t\alpha_{a,b,j}}\right\rangle
  =\left\langle\e^{\im t\sum_j\sum_{a,b}\alpha_{a,b,j}}\right\rangle
\end{equation}
This is exactly the characteristic function of the random variable $\sum_j\sum_{a,b}\alpha_{a,b,j}$, which, by the continuity theorem and Lyapunov's central limit theorem converges pointwise to the characteristic function of a standard normal random variable as $n$ grows, that is to $\e^{-\frac{t^2}{2}}$.

The fact that $\alpha_{a,b,j}$ are symmetric about zero means that $\phi_n(t)$ can still be written in the form
\begin{equation}
  \left\langle\frac{1}{2^n}\Tr\left(\prod_{k=1}^9\e^{\im tA_k}\e^{\im tB_k}\right)\right\rangle
\end{equation}
This then allows $|\psi_n(t)-\phi_n(t)|$ to be bounded by $t^2s_n^2\sqrt{n}\left(36\sqrt{2}+81\right)$ in the same way as before, the only change needed is to equation (\ref{variance}), where
\begin{equation}
  \left\langle\alpha_{a,b,j}\alpha_{a^\prime,b^\prime,j}\alpha_{a,b,j}\alpha_{a^\prime,b^\prime,j}\right\rangle
  =\left\langle\alpha_{a,b,j}^2\right\rangle\left\langle\alpha_{a^\prime,b^\prime,j}^2\right\rangle\leq s_n^4
\end{equation}
Then by the triangle inequality, and as $s_n^2=o\left(\frac{1}{\sqrt{n}}\right)$,
\begin{equation}
 \Big|\psi_n(t)-\e^{-\frac{t^2}{2}}\Big|\leq\Big|\psi_n(t)-\phi_n(t)\Big|+\Big|\phi_n(t)-\e^{-\frac{t^2}{2}}\Big|\to0
\end{equation}
as $n\to\infty$, which concludes the proof.
\end{proof}

\subsection{More general graphs of qubits}\label{geo}
The previous theorems are not restricted to the simple geometry of a ring of an even number of qubits.  In particular the proofs may be extended to a ring of an odd number of qubits, a chain of qubits or higher dimensional structures such as lattices of qubits.

Let $\Gamma_n$ be a sequence of simple graphs with $n$ labelled vertices and $e_n$ edges.  Furthermore, let each graph be $c$-colourable for some constant $c$, independent of $n$.  That is, there exist at least $c$ colours so that the edges of each graph $\Gamma_n$ can be coloured in a way that no vertex is connected to more than one edge of any one colour.

The graphs $\Gamma_n$ may then be used to define the sequence of random matrices,
\begin{equation}
  H_n^{(\Gamma)}=\sum_{(j,k)\in\Gamma_n}\sum_{a,b=1}^3\alpha_{a,b,j,k}\sigma_{  j  }^{(a)}\sigma_{  k  }^{(b)}
\end{equation}
where the $\alpha_{a,b,j,k}$ are independent normally distributed random variables with zero mean and variance $s_n^2=\frac{1}{9e_n}$, for each value of $n$ separately. The sum here is over all edges of $\Gamma_n$, where the edge connecting vertices $j$ and $k$ is labelled by $(j,k)$, with $j<k$ as a convention.

The following theorem deals with this modified ensemble:

\begin{theorem}
  The characteristic function $\phi_n(t)$, corresponding to the density of states probability measure for the ensemble above, converges pointwise to the characteristic function of a standard normal random variable as $n\to\infty$.
\end{theorem}
 
It again follows immediately from the continuity theorem that $\di\mu_n^{(DOS)}$ tends weakly to a standard normal distribution.  Also, universality can be proved for this ensemble by modifying the following proof in the same way as for Theorem \ref{universal}.

\begin{proof}
The proof follows exactly the same structure as that of Theorem \ref{DOS}.  A colouring of each graph $\Gamma_n$, with the colours $\mathcal{A}, \mathcal{B}, \mathcal{C},\dots$, may be used to define the operators $A,B,C,\dots$ and $A_k,B_k,C_k,\dots$ in an analogous way to equation (\ref{A}) and (\ref{A_i}).  For example
\begin{equation}
  A=\sum_{\genfrac{}{}{0pt}{}{(j,k)\in\Gamma_n}{(j,k)\sim\mathcal{A}}}\sum_{a,b=1}^3\alpha_{a,b,j,k}\sigma_{  j  }^{(a)}\sigma_{  k  }^{(b)}\qquad\qquad
  A_{3(b-1)+a}=\sum_{\genfrac{}{}{0pt}{}{(j,k)\in\Gamma_n}{(j,k)\sim\mathcal{A}}}\alpha_{a,b,j,k}\sigma_{  j  }^{(a)}\sigma_{k}^{(b)}
\end{equation}
where $(j,k)\sim\mathcal{A}$ (with $j<k$) denotes all edges with colour $\mathcal{A}$.  Now $H_n^{\Gamma}$ may be decomposed as
\begin{equation}\label{k}
  H_n^{\Gamma}=A+B+C+\dots=\sum_{k=1}^{9}\left(A_k+B_k+C_k+\dots\right)
\end{equation}
The characteristic function $\phi_n(t)$ is then defined analogously to equation (\ref{characteristic})
\begin{equation}
  \phi_n(t)=\e^{-\frac{t^2}{2}}=\left\langle\frac{1}{2^n}\Tr\left(\prod_{k=1}^{9}\e^{\im tA_k}\e^{\im tB_k}\e^{\im tC_k}\dots\right)\right\rangle
\end{equation}

The average of the commutator of any pair of the operators $A_k,B_k,C_k,\dots$ is bounded by $2s_n^2\sqrt{n}$, as in the proof of Theorem \ref{DOS}. The difference $|\psi_n(t)-\phi_n(t)|$ can then be bounded in an analogous way to that shown in the proof of Theorem \ref{DOS}. 
\end{proof}

\subsection{Local terms}\label{local}
We note briefly that the theorems established above remain valid if  local terms $\sum_{j=1}^n\sum_{a=1}^3\alpha_{a,0,j}\sigma_j^{(a)}$ are added to $H_n^{\Gamma}$.  These may be incorporated into the  preceding proof by adding a further operator $L$ to the list $A,B,C,\dots$, where
\begin{equation}
	L=\sum_{j=1}^n\sum_{a=1}^3\alpha_{a,0,j}\sigma_j^{(a)}\qquad\qquad
	L_k=\begin{cases}\sum_{j=1}^n\alpha_{k,0,j}\sigma_j^{(a)}\qquad&k=1,2,3\\0&\text{else}\end{cases}
\end{equation}

\section{Level spacing statistics}\label{level}
Ensembles of matrices with a similar structure to $H_n$ display an interesting variety of level spacing statistics.  Three ensembles will be investigated in this section: the first, the ensemble of matrices $H_n$ defined above, the second, this ensemble with the addition of local terms, and the third an ensemble of matrices that is translationally invariant along the chain and that also includes local terms.  These ensembles are defined by
\begin{align}
	H_n&=\sum_{j=1}^n\sum_{a,b=1}^3\alpha_{a,b,j}\sigma_{  j  }^{(a)}\sigma_{  j+1  }^{(b)},\,\,\quad\qquad\qquad\alpha_{a,b,j}\sim\mathcal{N}\left(0,\frac{1}{9n}\right)\text{ i.i.d}\nonumber\\
	H_n^{(local)}&=\sum_{j=1}^n\sum_{a=1}^3\sum_{b=0}^3\alpha_{a,b,j}\sigma_{  j  }^{(a)}\sigma_{  j+1  }^{(b)},\,\,\qquad\qquad\alpha_{a,b,j}\sim\mathcal{N}\left(0,\frac{1}{12n}\right)\text{ i.i.d}\nonumber\\
	H_n^{(inv)}&=\sum_{j=1}^n\sum_{a=1}^3\sum_{b=0}^3\alpha_{a,b}\sigma_{  j  }^{(a)}\sigma_{  j+1  }^{(b)},\quad\qquad\qquad\alpha_{a,b}\sim\mathcal{N}\left(0,\frac{1}{12n}\right)\text{ i.i.d}
\end{align}

\subsection{Numerical results}
The spacings between the ordered eigenvalues of the unfolded spectrum of $s$ random samples of each of the matrices $H_n$, $H_n^{(local)}$ and $H_n^{(inv)}$ were calculated numerically.  The unfolding is with respect to the ensemble's (numerical) density of states and is a rescaling of the eigenvalues such that this density becomes uniform on $[0,2^n]$.   Normalised histograms, averaged over the sampled matrices, are shown in Figures \ref{spacing1} and \ref{spacing2} for the values of $n$ indicated.

\begin{figure}
  \centering
  \includegraphics{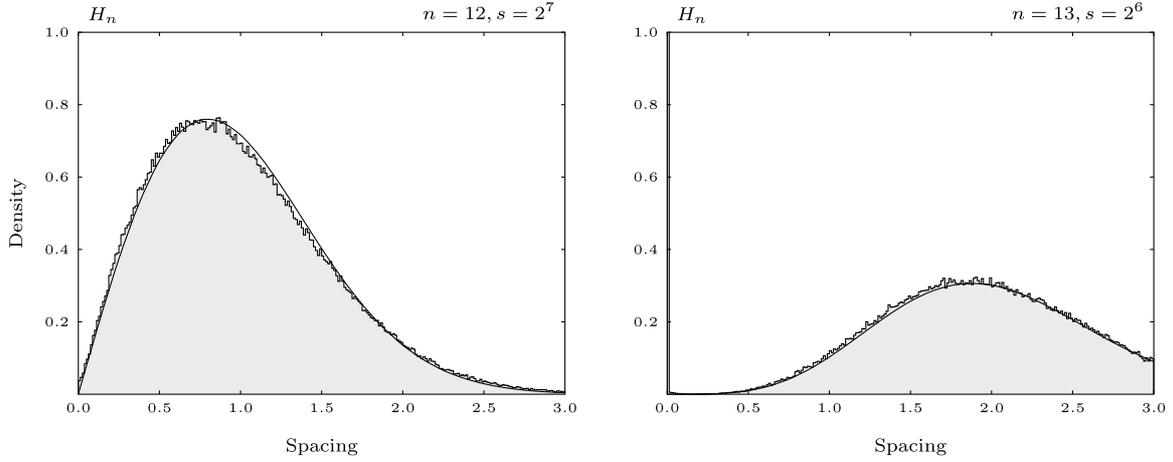}
	\caption{(Shaded histogram) The averaged normalised histograms of the nearest neighbour spacings of the unfolded and ordered eigenvalues over each of the $s$ matrices $H_n$ sampled.  The GOE (smooth line) limiting level spacing distribution is overlaid for $n=12$ with a rescaled GSE (smooth line) limiting level spacing distribution overlaid for $n=13$.}
	\label{spacing1}
\end{figure}

\begin{figure}
  \centering
  \includegraphics{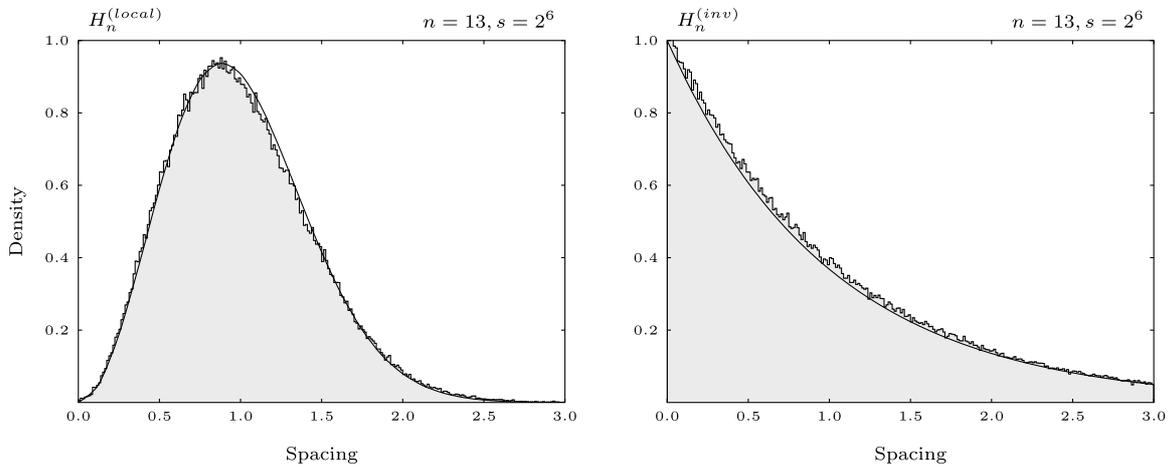}
	\caption{(Shaded histogram) The averaged normalised histograms of the nearest neighbour spacings of the unfolded and ordered eigenvalues over each of the $s$ matrices $H_n^{(local)}$ (left) and $H_n^{(inv)}$ (right) sampled.  The GUE (smooth line) limiting level spacing distributions is overlaid for $H_n^{(local)}$ and a Poisson (smooth line) distribution overlaid for $H_n^{(inv)}$.}
	\label{spacing2}
\end{figure}

For $H_n$ with $n=2,\dots,12$ even, there would appear to be convergence to the level spacing distribution of the Gaussian Orthogonal Ensemble (GOE).  For $n=3,\dots,13$ odd, all the matrices sampled from $H_n$ exhibited a degenerate spectrum leading to the peak at zero seen in Figure \ref{spacing1}.  The rest of the spacings however appear to converge to the Gaussian Symplectic Ensemble (GSE) form, rescaled to have mean $2$ and total area of $\frac{1}{2}$.  The samples of $H_n^{(local)}$ displayed a level spacing distribution tending to that of the Gaussian Unitary Ensemble (GUE) for all values of $n=2,\dots,13$.  For $H_n^{(inv)}$, level repulsion was no longer observed: the spacings distribution appears to converge to that of the Poisson distribution as $n$ increases from $2$ to $13$.

\subsection{Kramers degeneracy}\label{}
The degeneracies seen in the spectrum of instances of $H_{2m+1}$  are examples of Kramers degeneracy.  The following two lemmas explain their existence:
\begin{lemma}[Pseudo time reversal symmetry]
  The matrix ${H}_n$ satisfies
  \begin{equation}
    SH_n=\overline{H_n}S
  \end{equation}
  where $S=S^\dagger=S^{-1}={\sigma^{(2)}}^{\otimes n}$.  The bar denotes complex conjugation of the matrix (or later vector) elements in the standerd basis, in which the Pauli matrices in (\ref{Pauli}) are expressed.
\end{lemma}

\begin{proof}
For $a=1,2,3$, $\sigma^{(2)}\sigma^{(a)}\sigma^{(2)}=-\overline{\sigma^{(a)}}$, so that $S{H}_nS=\overline{{H}_n}$ as non-identity Pauli matrices only occur in pairs in ${H}_n$ and all the coefficients $\alpha_{a,b,j}$ are real.
\end{proof}

The local terms in $H_n^{(local)}$ and $H_n^{(inv)}$ break this symmetry.

This proposition could be reformulated in terms of an anti-unitary time reversal operator $\Theta=KS$, where the action of $K$ is to perform the necessary complex conjugation so that $\Theta H_n=H_n\Theta$.  

\begin{lemma}[Kramers degeneracy]
  For odd $n$, the matrix ${H}_n$ has, at least, doubly degenerate eigenvalues.
\end{lemma}

\begin{proof}
The proof follows the standard arguments for showing a Kramers degeneracy.  Let $|\psi\rangle$ be an eigenstate of ${H}_n$ with an eigenvalue of $\lambda$.  As $S{H}_n=\overline{{H}_n}S$, it follows that
\begin{equation}
  \overline{{H}_n}\Big(S|\psi\rangle\Big)=S {H}_n|\psi\rangle=\lambda\Big(S|\psi\rangle\Big)
\end{equation}
Upon taking the complex conjugate of this equation it is seen that $\overline{S|\psi\rangle}$ is an eigenstate of ${H}_n$ with eigenvalue $\lambda$.

Now the inner product of $\overline{S|\psi\rangle}$ and $|\psi\rangle$ will be calculated.  As $S^\dagger S=I$ and $S\overline{S}=(-I)^n$ it follows that
\begin{equation}
  \langle\psi|\overline{S|\psi\rangle}=\langle\psi|S^\dagger S\overline{S|\psi\rangle}=(-1)^n\langle\psi|S^\dagger\overline{|\psi\rangle}
\end{equation}
As $\langle\psi|S^\dagger\overline{|\psi\rangle}=\overline{\overline{\langle\psi|}S|\psi\rangle}=\langle\psi|\overline{S|\psi\rangle}$ by the definition of the complex conjugate of a matrix, it is concluded that $\langle\psi|\overline{S|\psi\rangle}=0$ for all odd $n$.  Hence, $|\psi\rangle$ and $\overline{S|\psi\rangle}$ are orthogonal eigenstates of ${H}_n$, both with the eigenvalue $\lambda$, for all odd $n$.
\end{proof}

Dyson's threefold way leads one to expect GOE statistics when an anti-unitary symmetry exists that squares to +1, GSE statistics when this symmetry squares to -1 (i.e.~when there is a Kramers degeneracy), and GUE statistics when the symmetry is broken.  Our findings are consistent with this.  Note, however, that the matrices we are here working with are extremely sparse compared with the matrices that form the standard random matrix ensembles (cf.~also the number of free parameters compared to the matrix size and the fact that the mean density of states does not match any of the standard random matrix forms).  In the ensemble of translationally invariant Hamiltonians the appearance of Poisson statistics is consistent with the presence of the geometric symmetry.

\section{Discussion and open questions}\label{dis}
\subsection{Speed of convergence}
Samples of 
\begin{equation}
	H_n^{(JW)}=\frac{1}{\sqrt{C}}\sum_{j=1}^{n-1}\sum_{a,b=1}^2\hat{\alpha}_{a,b,j}\sigma_{j}^{(a)}\sigma_{  j +1 }^{(b)}+\frac{1}{\sqrt{C}}\sum_{j=1}^{n}\hat{\alpha}_{3,0,j}\sigma_j^{(3)}
\end{equation}
where the coefficients $\hat{\alpha}_{a,b,j}$ are standard normal random variables and $C$ is the sum of their squares, provide nontrivial instances of matrices from the ensemble $H_n^{(local)}$ to which the modified versions of Theorem \ref{DOS} apply.  The eigenvalues can be computed by numerically diagonalising the $2n\times 2n$ matrix arising from applying the Jordan-Wigner transform to $H_n^{(JW)}$ \cite{Nielsen2005}.  Pointwise convergence of the integrated density of states to the standard normal distribution numerically appears to be at a rate on the order of $\frac{1}{n}$ for values of $n$ up to $32$, see Figure \ref{DOSn}.  This is close to the bound on the rate given in the theorem.  Under the general conditions stated our bound therefore appears to be close to being sharp.

\begin{figure}
	\centering
  \includegraphics{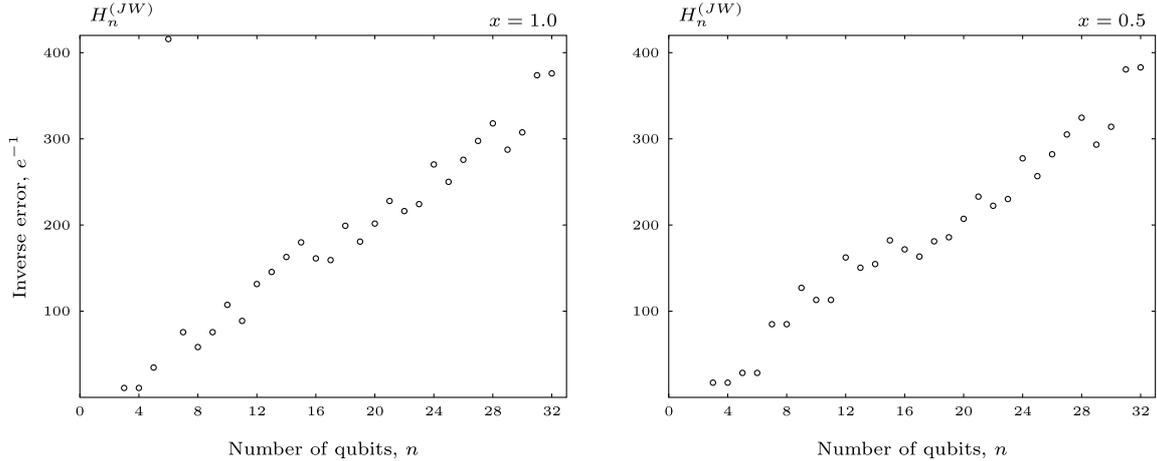}
	\caption{The circles give the values of $e(x, n)^{-1}=\left|\frac{1}{2^n}\sum_{k}\int_{-\infty}^{x}\delta(\lambda-\lambda_{k})\di\lambda-\frac{1}{\sqrt{2\pi}}\int_{-\infty}^{x}\e^{-\frac{\lambda^2}{2}}\di\lambda\right|^{-1}$ against chain length, $n$, for $x=1$ (left) and $x=0.5$ (right), where $\lambda_k$ are the eigenvalues of a random instance of $H_n^{(JW)}$.   A roughly linear relationship is seen to emerge for large $n$ with similar results holding for other values of $x$ and all other instances of $H_n^{(JW)}$ tested.}
	\label{DOSn}
\end{figure}

The average unfolded (scaled to have a unit mean) distribution of spacings between the ordered eigenvalues of $H_n^{(JW)}$ numerically tends towards the Poisson distribution as $n$ increases.  This is in contrast to the average unfolded spacing distribution of $H_n^{(local)}$ which, numerically, tends to the GOE spacing distribution.  It is consistent with the fact that these systems are integrable.

It may be that the rate of convergence of the average density of states measure for $H_n$ and $H_n^{(local)}$ is on the scale of $\frac{1}{2^n}$, the mean level spacing, and not on the scale $\frac{1}{n}$ observed for $H_n^{(JW)}$, in order for random matrix statistics to emerge.  This remains an open question.  

Convergence of the density of states probability measure to the standard normal distribution is also numerically seen for $H_n^{(inv)}$.  The machinery of Theorem \ref{DOS} is not immediately applicable to this situation though as the coefficients along the chain are not statistically independent.  We shall address this further in \cite{KLW2014}.

Proving that the spectral statistics of $H_n^{(inv)}$ coincide with those of the standard Gaussian Ensembles (or more generally, the Wigner and invariant ensembles) also remains an open problem.  This would appear to be difficult, given the relatively small number of free parameters compared to the matrix size.

\section*{Acknowledgements} We are grateful to Eugene Bogomolny and Leonid Pastur for helpful discussions, and to Anna Maltsev for valuable advice and comments.

\section*{Appendices}
\appendix
\section{Numerical results for the density of states}\label{numDOS}
Figure \ref{numerics} shows the results of a numerical simulation of the density of states probability measure for the ensemble of matrices $H_n$.  Here, $s$ matrices were numerically sampled from the ensemble and diagonalised.  The number of eigenvalues falling in each of $240$ consecutive intervals of equal width between $-3$ and $3$ were then counted and the normalised histograms in Figure \ref{numerics} produced.  A strong resemblance to the standard normal distribution is seen.  Repeating this procedure for $n=2,\dots,13$ produces a sequence of curves which appears to converge to the standard normal distribution.

\begin{figure}
  \centering
  \includegraphics{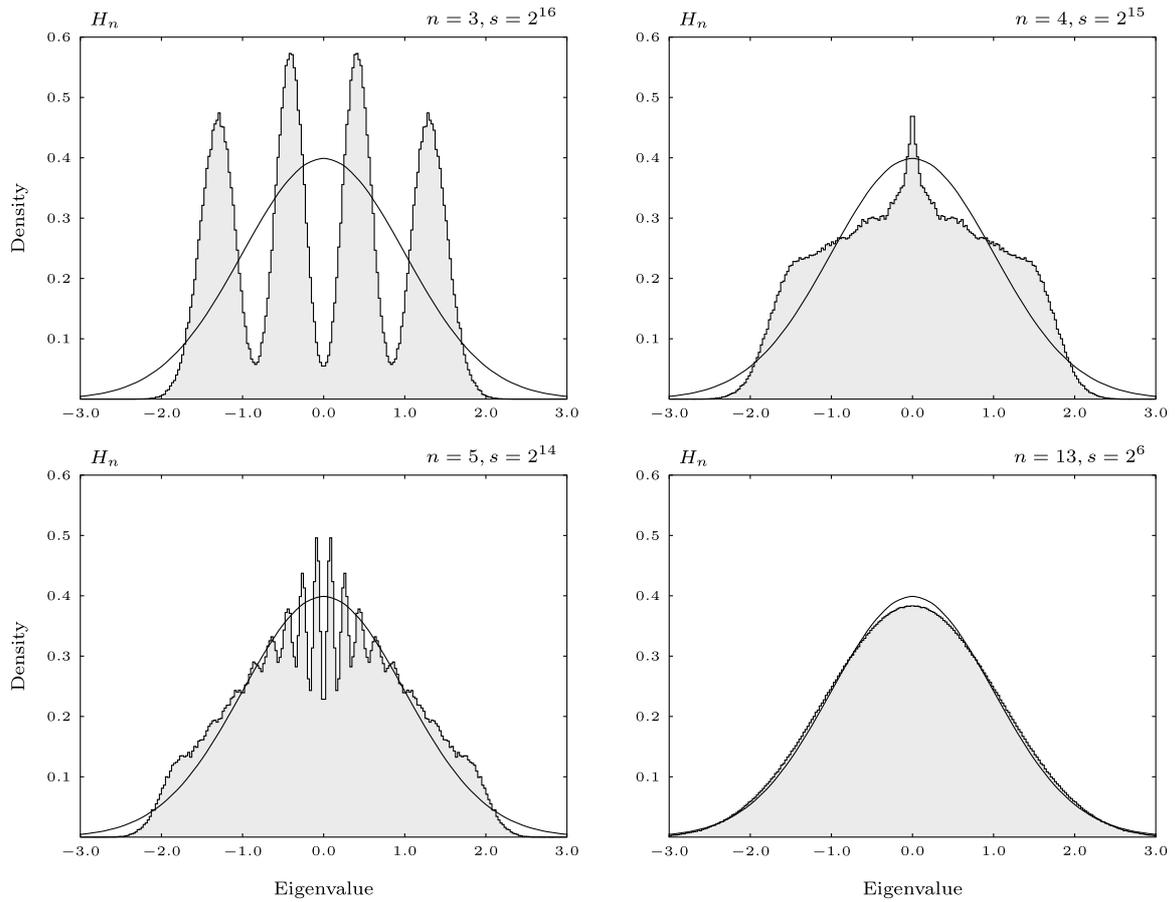}
  \caption{(Shaded histogram) The average, normalised, density of states histograms over $s$ random samples of the Hamiltonian
  $H_n$ for $n=3,4,5,13$.  (Smooth line) The probability density function for a standard normal random variable.} 
 \label{numerics}
\end{figure}

\section{Integral identity}\label{Int}
For any $2^n\times2^n$ Hermitian matrices $X$ and $Y$ and real parameter $t$ the following identity will be shown
\begin{align}
  \e^{\im t(X+Y)}-\e^{\im tX}\e^{\im tY}=\int_0^1\int_0^1t^2s\e^{\im (1-s)t(X+Y)}\e^{\im (1-r)stX}[X,Y]\e^{\im rstX}\e^{\im stY}\di r\di s
\end{align}
First, by the fundamental theorem of calculus
\begin{equation}
  \e^{\im t(X+Y)}-\e^{\im tX}\e^{\im tY}=-\int_0^1\frac{\partial}{\partial s}\left(\e^{\im (1-s)t(X+Y)}\e^{\im stX}\e^{\im stY}\right)\di s
\end{equation}
which by computing the derivative gives
\begin{equation}
  \e^{\im t(X+Y)}-\e^{\im tX}\e^{\im tY}=\im t\int_0^1\e^{\im (1-s)t(X+Y)}\left[Y,\e^{\im stX}\right]\e^{\im stY}\di s
\end{equation}
Similarly, by the fundamental theorem of calculus
\begin{equation}
  \left[Y,\e^{\im stX}\right]=\int_0^1\frac{\partial}{\partial r}\left(\e^{\im(1-r)stX}Y\e^{\im rstX}\right)\di r
\end{equation}
which by computing the derivative gives
\begin{equation}
  \left[Y,\e^{\im sX}\right]=-\im st\int_0^1\e^{\im(1-r)stX}[X,Y]\e^{\im rstX}\di r
\end{equation}
Therefore combining the above results gives the claimed identity.

\footnotesize
\bibliographystyle{unsrt}
\bibliography{EnsembleDOS}

\end{document}